\documentclass[11pt]{article}
\usepackage{a4wide,xspace}%,times}
\usepackage{graphicx}
\usepackage{amsmath,upgreek}
\usepackage{amsfonts}
\usepackage{amssymb}
\usepackage{epstopdf}

\newtheorem{theorem}{Theorem}

\newtheorem{lemma}[theorem]{Lemma}

\newtheorem{remark}[theorem]{Remark}

\newcommand{\Xomit}[1]{ }
\newenvironment{proof}[1][Proof]{\textbf{#1.} }{\ \rule{0.5em}{0.5em}}

\mathchardef\mhyphen="2D

\newcommand{\eps}{\upvarepsilon}

\newcommand{\gam}{\upgamma}

\begin{document}

\title{A new lower bound for classic online bin packing}

%\title{Lower bounds for several online variants of bin packing\footnote{An extended abstract version appears in the Proceedings of WAOA2017. \newline
%J. Balogh was supported by the project ``Integrated program for
%training new generation of scientists in the fields of computer
%science'', no. EFOP-3.6.3-VEKOP-16-2017-0002. J. B\'ek\'esi was
%supported by the EU-funded Hungarian grant
%EFOP-3.6.2-16-2017-00015. Gy. D\'{o}sa was supported by Szechenyi
%2020 under the EFOP-3.6.1-16-2016-00015 and by National Research,
%Development and Innovation Office  NKFIH under the grant SNN
%116095. L. Epstein and A. Levin were partially supported by a
%grant from GIF - the German-Israeli Foundation for Scientific
%Research and Development (grant number I-1366-407.6/2016).}}

\date{}

\author{J\'anos Balogh \thanks{Department of Applied Informatics, Gyula Juh\'asz Faculty of Education,
     University of Szeged, Hungary. \texttt{balogh@jgypk.u-szeged.hu}} \and J\'ozsef B\'ek\'esi \thanks{Department of Applied Informatics, Gyula Juh\'asz Faculty of Education,
     University of Szeged, Hungary. \texttt{bekesi@jgypk.u-szeged.hu}} \and Gy\"{o}rgy
D\'{o}sa\thanks{Department of Mathematics, University of Pannonia,
Veszpr\'em, Hungary, \texttt{dosagy@almos.vein.hu}.} \and Leah
Epstein\thanks{ Department of Mathematics, University of Haifa,
Haifa, Israel. \texttt{lea@math.haifa.ac.il}. } \and Asaf
Levin\thanks{Faculty of Industrial Engineering and Management, The
Technion, Haifa, Israel. \texttt{levinas@ie.technion.ac.il.}}}

%\vspace{-0.5cm}

\maketitle

\begin{abstract}
We improve the lower bound on the asymptotic competitive ratio of
any online algorithm for bin packing to above $1.54278$. We
demonstrate for the first time the advantage of branching and the
applicability of full adaptivity in the design of lower bounds for
the classic online bin packing problem. We apply a new method for
weight based analysis, which is usually applied only in proofs of
upper bounds. The values of previous lower bounds were
approximately $1.5401$ and $1.5403$.
\end{abstract}

\section{Introduction}
The bin packing problem \cite{U71,J73,Gon32} is a well-studied
combinatorial optimization problem with origins in data storage
and cutting stock. The input consists of items of rational sizes
in $(0,1]$, where the goal is to split or pack them into
partitions called bins, such that the total size of items for every bin cannot
exceed $1$. The online bin packing problem \cite{CsiWoe98,Gon32}
is its variant where items are presented one by one, and the
algorithm assigns each item to a bin before it can see the next
item.

For an algorithm $A$ and an input $I$, let $A(I)$ be the cost
(number of bins) used by $A$ for $I$. The algorithm $A$ can be an
online or offline algorithm, and it can also be an optimal offline
algorithm OPT. The absolute competitive ratio of algorithm $A$ for
input $I$ is the ratio between $A(I)$ and $OPT(I)$. The absolute
competitive ratio of $A$ is the worst-case (or supremum) absolute
competitive ratio over all inputs. Given an integer $N$, we can
consider the worst-case absolute competitive ratio over inputs
where $OPT(I)$ is not smaller than $N$. Taking this sequence and
letting $N$ grow to infinity, the limit is the asymptotic
competitive ratio of $A$. This measure is the standard one for
analysis of the bin packing problem, and it is considered to be
more meaningful than the absolute ratio (which is affected by very
small inputs).

The current best online algorithm with respect to the asymptotic
competitive ratio has an asymptotic competitive ratio no larger
than $1.57829$ \cite{DBLP:journals/corr/BaloghBDEL17}, which was
found recently by development of new methods of analysis. Previous
results were achieved via a sequence of improvements
\cite{J74,JDUGG74,Yao80A,LeeLee85,RaBrLL89,Seiden02J,HvS16}. In
this work, we consider the other standard aspect of the online
problem, namely, of establishing lower bounds on the asymptotic
competitive ratio that can be achieved by online algorithms.

The first lower bound on the asymptotic competitive ratio was
found by Yao \cite{Yao80A}, and it uses an input with at most
three types of items: $\frac 17+\eps$, $\frac 13+\eps$, and $\frac
12+\eps$ (where $\eps>0$ is sufficiently small). For this input,
if the entire input is presented, every bin of an optimal solution
has one item of each type (and otherwise there are larger numbers
of items in a bin, but all bins are still packed identically). It
is possible to start the sequence with smaller items, for example,
it can be started with $\frac 1{1807}+\eps$ and then $\frac
1{43}+\eps$, which increases the result. This was discovered by
Brown and Liang (independently)
\cite{liang1980lower,brown1979lower}, who showed a lower bound of
1.53635.  Van Vliet \cite{Vliet92} found an interesting method of
analysis and showed that the same approach (the same sequence with
additional items) gives in fact a lower bound of 1.5401474.
Finally, Balogh, B{\'e}k{\'e}si, and Galambos \cite{balogh2012new}
showed that the greedy sequence above is actually not the best one
among sequences with batches of identical items, and proved a
lower bound of $248/161 \approx 1.5403726$ (see also \cite{BDE}
for an alternative proof). Their sequence starts with decreasing
powers of $\frac 17$ plus epsilon (it can be started with items
complementing the other items to $1$ but it does not change the
bound), and after $\frac{1}{49}+\eps$ the other items are exactly
those used by Yao \cite{Yao80A}.  This result of
\cite{balogh2012new} is the previously best known lower bound.

One drawback of the previous lower bounds is that while the exact
input was not determined in advance, the set of sizes used for it
was determined prior to the action of the algorithm by the input
provider and it was known to the algorithm. Moreover, for classic
bin packing, in all previously designed lower bound inputs, sizes
of items were slightly larger than a reciprocal of an integer, and
optimal solutions consisted of bins with identical packing
patterns. The possible item sizes and numbers of items were known
to the algorithm, but the stopping point of the input was unknown,
and it was based on the action of the algorithm. It seemed
unlikely that such examples are indeed the worst-case examples. We
show here that different methods for proving lower bounds and new
approaches to sizes of items give an improved lower bound.

\paragraph{New features of our work.}
Previous lower bound constructions for standard bin packing were
defined for inputs without branching. Those are inputs where the
possible inputs differ only by their stopping points. Here, we use
an input with branching, which makes the analysis harder, as those
branches are related (the additional items may use the same
existing bins in addition to new bins), but at most one of them
will be presented eventually. It is notable that branching was
used to design an improved lower bound for the case where the
input consists of three batches \cite{BBDGT15} (where for each one
of the batches, all items are presented at once), but it was
unknown whether it can be used to design improved lower bounds for
standard online bin packing. That is, it was unknown if the impact
of branching in \cite{BBDGT15} is similar to one additional batch
or if it gives the adversary more power that can be used in the
general settings as well.

It was also not known whether one can exploit methods of
constructing fully adaptive inputs, where in some parts of the
input every item size is based precisely on previous decisions of
the algorithm. Such results were previously proved for online bin
packing with cardinality constraints, where (in addition to the
constraint on the total size) every bin is limited to containing
$k$ items, for a fixed parameter $k \geq 2$
\cite{Blitz,BCKK04,FK13,BBDEL_ESA}. Thus, in addition to branching
we will use the following theorem proved in \cite{BBDEL_ESA} (see
the construction in Section 3.1 and Corollary 3  in
\cite{BBDEL_ESA}).

\begin{theorem}\label{cnstrct}
Let $N \geq 1$ be a large positive integer and let $k \geq 2$ be
an integer. Assume that we are given an arbitrary deterministic
online algorithm for a variant of bin packing and a binary
condition $Con_1$ on the possible behavior of an online algorithm
for one item (on the way that the item is packed). An adversary is
able to construct a sequence of values $a_i$ ($1 \leq i \leq N$)
such that for any $i$, $a_i \in \left( k^{- 2^{N+3}}, k^{-2^{N+2}}
\right)$, and in particular $a_i \in \left(0,\frac{1}{k^4}\right)$
(defining item sizes is done using a given linear function of the
values $a_i$), such that for any item $i_1$ satisfying $Con_1$ and
any item $i_2$ not satisfying $Con_1$, it holds that
$\frac{a_{i_2}}{a_{i_1}} > k$.
\end{theorem}

Examples for the condition $Con_1$ can be the following: ``the
item is packed as a first item of its bin'', ``the item is packed
into a non-empty bin'', ``the item is packed into a bin already
containing an item of size above $\frac 12$'', etc. Here, the
condition $Con_1$ will be that the item is not packed into an
empty bin (or a new bin).

Our method of analysis is based on a new type of a weighting
function. This kind of analysis is often used for analyzing bin
packing algorithms, that is, for upper bounds. It was used for
lower bounds \cite{BDE} and by van Vliet \cite{Vliet92} (where the
term weight is not used, and the values given to items are based
on the dual linear program, but the specific kind of dual
variables and their usage can be adapted to a weighting function).
However, those weights were defined for inputs without branching
and we extend the use of these weights for inputs with branching
for the first time, which adds technical challenges to our work
also in the analysis. The advantage of weights is that we do not
need to test all packing patterns of an algorithm, whose number
can be very large, and thus we obtain a complete and verifiable
proof with much smaller number of cases than that of pattern based
proofs (see for example \cite{FK13}).

\section{The input}

Let $t \geq 3$ be an integer, let $\eps >0 $ be small constant,
let $M$ be a large integer and let $N = M \cdot 42^{t}$ ($N$ is a
large integer divisible by $6\cdot 7^t$). The condition on $\eps$
is: $\eps < \frac{1}{(2058)^{t}}$.

Given a specific algorithm ALG, we will analyze it for the set of
inputs defined here, where the input depends on the actions of ALG
both with respect to stopping the input, but also some of the
sizes will be based on the exact action of ALG, and on the
previously presented items and their number.

Let $C_t = \frac{1}{6\cdot 7^{t-1}}-294\eps$, and for $2 \leq j
\leq t-1$, let $C_j=\frac{1+28\eps}{7^j}$. The input starts with
batches of $N$ items of the sizes $C_j$, for every
$j=t,t-1,\ldots,2$, where the input may be stopped after each one
of these batches. An item of size $C_j$ is called a $C_j$--item.

Afterwards, there are $N$ items called $A$--items. The sizes of
$A$--items will be all strictly larger than $\frac{1+\eps}7$ but
strictly smaller than $\frac{1+2\eps}7$. Any $A$--item packed as a
first item into a bin will be called a large $A$--item, and any
other $A$--items will be called a small $A$--item. During the
construction, based on the actions of the algorithm, we will
ensure that for any large $A$--item, its difference from
$\frac{1+\eps}7$ is larger by a factor of more than $4$ than the
difference from $\frac{1+\eps}7$ of any small $A$--item. The
details of attaining this property are given below (see Lemma
\ref{lem2}).

Let $\gam>0$ be such that the size of every small $A$--item is at
most $\frac{1+\eps +\gam}7$ while the size of every large
$A$--item is above $\frac{1+\eps +4\gam}7$ (where $\gam <
\frac{\eps}4$). The input may be stopped after $A$--items are
introduced (the number of $A$ items is $N$ no matter how many of
them are small and how many are large). Let $n_L$ denote the
number of large $A$--items, and therefore there are $N-n_L$ small
$A$--items. Even though the $A$--items will have different sizes
and they cannot be presented at once to the algorithm, we see them
as one batch.

If the input is not stopped after the arrival of $A$--items, there
are three options to continue the input (i.e., we use branching at
this point). In order to define the three options, we first define
the following five items types. A $B_{11}$--item has size $\frac
{1+2\eps}2$. A $B_{21}$--item has size $\frac {1+\eps}3$ and a
$B_{22}$--item has size $\frac {1+\eps}2$. A $B_{31}$--item has
size $\frac{5-2\eps-3\gam}{14}$ and a $B_{32}$--item has size
$\frac{7+\gam}{14}=\frac {1}2+\frac{\gam}{14}< \frac 12 +
\frac{\eps}{56}$ (this size is above $\frac 12$).

The first option to continue is with $B_{11}$--items, such that a
batch of $\frac N3$ such items arrive. The second option is with a
batch of $B_{21}$--items, possibly followed by a batch of
$B_{22}$--items. In this option, the number of items of each batch
is $N$. The third option is that a batch of $B_{31}$--items
arrive, possibly followed by a batch of $B_{32}$--items. In the
last case, we define the numbers of items based on $n_L$ as
follows. The number of $B_{31}$--items (if they are presented) is
$n_{31}=\frac{7N - 7n_L}6$. The number of $B_{32}$--items (if they
are presented) is $n_{32}=\frac{7N -5n_L}6$.  This concludes the
description of the input (see Figure \ref{thetree} for an
illustration).

\begin{figure} [h!]
\vspace{0.65cm} \hspace{0.9in}
\includegraphics[angle=270,width=0.75\textwidth]{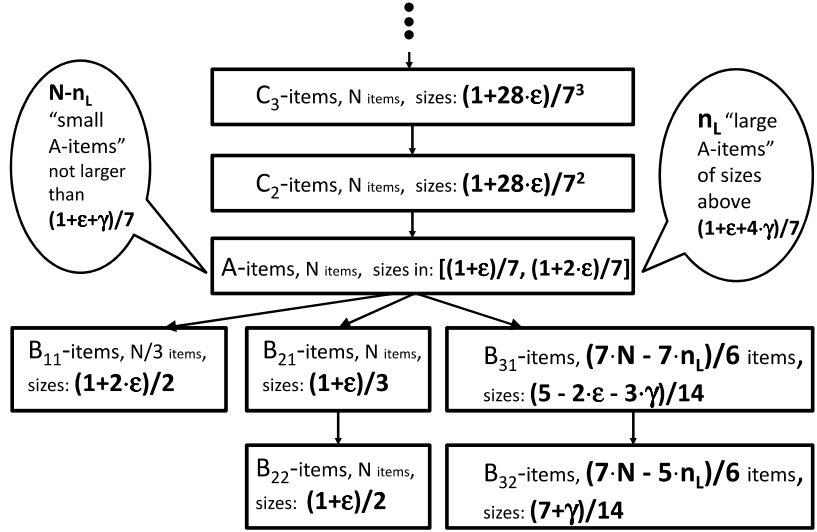}
\vspace{-0.87cm} \caption{An illustration of the input. Every box
contains a set of items, and the input may be stopped after
presenting the items of any box. In cases with branching, at most
one path is selected, and any such path may be presented as an
input. \label{thetree}}
\end{figure}

We conclude this section by showing that indeed we can construct
the batch (or subsequence) of $A$--items satisfying the required
properties.

\begin{lemma}\label{lem2}
The sizes of $A$--items can be constructed as described.
\end{lemma}
\begin{proof}
We use Theorem \ref{cnstrct}. Condition $Con_1$ is that the item
is packed into a bin that already contains at least one item (this
item may be of a previous batch of items). Let $k=\lceil \frac
1{\eps} \rceil$. The items sizes are $\frac {1+\eps+a_i}{7}$. We
find that all item sizes are in $(\frac {1+\eps}{7},\frac
{1+2\eps}{7})$. We also have for two items of sizes $\frac
{1+\eps+a_{i_1}}{7}$ and $\frac{1+\eps+a_{i_2}}{7}$, where the
second item does not satisfy $Con_1$ while the first one satisfies
$Con_1$ that $\frac{a_{i_2}}{a_{i_1}}>k$. Let $\gam$ be the
maximum size of any value $a_i$ of an item satisfying $Con_1$.
Then, we have $\frac {1+\eps+ a_{i_1}}{7} \leq \frac
{1+\eps+\gam}{7}$ and $\frac {1+\eps+ a_{i_2}}{7} \geq \frac
{1+\eps+k \gam}{7}
> \frac {1+\eps+4 \gam}{7}$, as required.
\end{proof}

In order to give some motivation regarding the sizes of items,
note that by $\eps < \frac{1}{(2058)^{t}}$, we have
$\frac{5-2\eps-3\gam}{14} \geq \frac 5{14}-\frac{2.75}{14}\cdot
\eps > 0.35714$, while $\frac{1+\eps}3 < 0.33334$.

\section{Bounds on the optimal costs}

We find upper bounds on optimal costs. We denote the optimal cost
after the batch of items of sizes $C_j$ is presented by $OPT_j$
(for $j \geq 2$). Similarly, we denote an optimal cost after the
batch of $A$--items by $OPT_1$.

\begin{lemma}
For $t\geq j \geq 2$, we have $OPT_j \leq \frac{N}{6 \cdot
7^{j-1}}$, and $OPT_1 \leq \frac{N}{6}$. Furthermore, let $j\geq
1$, then the total size of one item of each batch up to the batch
of $C_j$--items (if $j\geq 2$) or up to the batch of $A$--items
(if $j=1$) is at most $\frac{1}{6\cdot 7^{j-1}} - 293\eps$.
\end{lemma}
\begin{proof}
First, consider $j\geq 2$. The total size of $t-j+1$ items,
each of a different size out of $C_t,C_{t-1},\ldots,C_j$ is
$C_t+\sum_{i=j}^{t-1} C_j=\frac{1}{6\cdot
7^{t-1}}-294\eps+\sum_{i=j}^{t-1}\frac{1+28\eps}{7^j}=\frac{1}{6\cdot
7^{t-1}}-294\eps+\frac{1+28\eps}{7^j}\sum_{i=j}^{t-1}\frac{1}{7^{i-j}}
=\frac{1}{6\cdot
7^{t-1}}-294\eps+\frac{1+28\eps}{7^j}\cdot\frac{7-1/7^{t-j-1}}6=\frac{1}{6\cdot
7^{t-1}}-294\eps+(1+28\eps)\frac{1/7^{j-1}-1/7^{t-1}}6 <
\frac{1}{6\cdot 7^{j-1}} - 293\eps$, as $\sum_{i=j}^{t-1}
\frac{1}{7^{i-j}}=\sum_{i=0}^{t-j-1}\frac{1}{7^{i}}=\frac{1-1/7^{t-j}}{6/7}=
\frac{7-1/7^{t-j-1}}6$. Thus, it is possible to pack $6\cdot
7^{j-1}$ items of each size into every bin and get a feasible solution with
$\frac{N}{6\cdot 7^{j-1}}$ bins, so $OPT_j \leq \frac{N}{6\cdot
7^{j-1}}$.

A similar bound can be used for the input up to the batch of
$A$--items as well. In this case the total size of one item of
each size $C_j$ together with one $A$--item (small or large) is at
most $\frac{1}{6\cdot
7^{t-1}}-294\eps+(1+28\eps)\frac{1/7-1/7^{t-1}}6+\frac {1+2\eps}7
 <\frac 16 - 293\eps$. Thus, $OPT_1 \leq \frac{N}{6}$ (by packing six items of each batch into every bin).
\end{proof}

We let $OPT_{11}$, $OPT_{21}$, $OPT_{22}$, $OPT_{31}$, $OPT_{32}$,
denote costs of optimal solutions for the inputs after the batches
of $B_{11}$--items, $B_{21}$--items, $B_{22}$--items,
$B_{31}$--items, and $B_{32}$--items were presented, respectively.
In the next lemma we present upper bounds on these optimal costs.

\begin{lemma}
We have $OPT_{11} \leq \frac{N}3$, $OPT_{21} \leq \frac{N}2$,  $OPT_{22} \leq {N}$, $OPT_{31}
\leq \frac{7N-5n_L}{12}$, and $OPT_{32} \leq
\frac{7N  -5n_L}6$.
\end{lemma}
\begin{proof}
We have $OPT_{11} \leq \frac{N}3$, as it is possible to pack three
items of each batch and one $B_{11}$--item into each bin of a
feasible solution, since their total sizes are below $3(\frac
16-293\eps)+\frac {1+2\eps}2<1$.

We have $OPT_{21} \leq \frac{N}2$, as it is possible to pack two
items of each batch and two $B_{21}$--items into each bin of a
feasible solution, since their total sizes are below $2(\frac
16-293\eps)+2\frac {1+\eps}3<1$.

We have $OPT_{22} \leq {N}$, as it is possible to pack one item of
each batch, one $B_{11}$--item and one $B_{21}$--item into each
bin, since their total sizes are below $(\frac 16-293\eps)+\frac
{1+\eps}3+\frac {1+\eps}2<1$.

If the third option for continuing the input is used, we define
solutions as follows. The first solution is for the input up to
$B_{31}$--items. There are bins with six large $A$--items, and six
items of each one of the  preceding batches, there are bins with
two $B_{31}$--items and two small $A$--items, and finally there
are bins with two $B_{31}$--items and $12$ items of each of the
sizes $C_t,C_{t-1},\ldots,C_2$. The numbers of bins of the three
types are $\frac{n_L}6$, $\frac{N-n_L}2$, and $\frac{N-n_L}{12}$,
respectively. We next argue that this is a feasible solution. The
number of $B_{31}$--items that are packed is $7\frac{N-n_L}{6}$,
the number of large $A$--items that are packed is $n_L$, the
number of small $A$--items that are packed is $N-n_L$, and for
every $j$ ($2 \leq j \leq t$), the number of $C_j$--items is
$n_L+(N-n_L)=N$, so all items are packed. The bins are valid as
the total size of items packed into a bin is at most
$1-\frac{\gam}{7} < 1$, as we show now.  For a bin of the first
type, the total size of items  is at most
$6(\frac{1+2\eps}7)+6(\frac{1}{42} - 293\eps)<1-290\eps \leq
1-\frac{\gam}7$. The total size of items packed into a bin of the
second type is at most $2(\frac{1+\eps
+\gam}7)+2(\frac{5-2\eps-3\gam}{14})=1-\frac{\gam}7$. Finally, the
total size of items packed into a bin of the third type is at most
$12(\frac{1}{42} -
293\eps)+2(\frac{5-2\eps-3\gam}{14})<1-290\eps\leq
1-\frac{\gam}7$. Thus, $OPT_{31} \leq \frac{n_L}6 + \frac{N-n_L}2
+ \frac{N-n_L}{12} = \frac{7N-5n_L}{12}$.

The second solution is for the input up to $B_{32}$--items. Every
bin has one $B_{32}$--item. The other contents are as follows.
There are bins with three large $A$--items, and three items of
each one of the preceding batches, there are bins with one
$B_{31}$--item and one small $A$--item, and finally there are bins
with one $B_{31}$--item and $6$ items of each of the sizes
$C_t,C_{t-1},\ldots,C_2$. Those are halves of the contents of the
bins in the case that $B_{32}$--items do not arrive, and the total
sizes of such halves are at most $\frac{1-\frac{\gam}7}2$. The
total size packed into each bin is at most
$(1-\frac{\gam}7)/2+(\frac{7+\gam}{14})=1$. Thus, $OPT_{32} \leq
\frac{7N  -5n_L}6$.
\end{proof}

We next prove that the optimal costs are at least $M$ (for all
possible inputs). We have $C_t = \frac{1}{6\cdot 7^{t-1}}-294\eps
> \frac{1}{6\cdot 7^{t-1}}-294\cdot \frac{1}{2058^t} >
\frac{1}{6\cdot 7^{t-1}+1}$, as $\frac{1}{6\cdot
7^{t-1}}-\frac{1}{6\cdot 7^{t-1}+1}=\frac{1}{6\cdot 7^{t-1}(6\cdot
7^{t-1}+1)}$ while $294\cdot
\frac{1}{2058^t}=\frac{1}{6^{t-1}\cdot 7^{3t-2}}$, and $6\cdot
7^{t-1}(6\cdot 7^{t-1}+1) < 6^{t-1}\cdot 7^{t-1} \cdot 7^{t} <
6^{t-1}\cdot 7^{3t-2}$ by $t \geq 2$. Thus, as all inputs contain
the first batch of $C_t$--items, and every bin has at most $6\cdot
7^{t-1}$ such items, we get that an optimal solution has at least
$\frac{N}{6\cdot 7^{t-1}}>M$ bins.

\section{An analysis using weights}
In this section we provide a complete analytic proof of the
claimed lower bound that we establish using our construction.  In
fact we verified the tightness of our analysis (for this
construction) by solving a mathematical program for some very
small values of $t$.  Our analytic proof is based on assigning
weights to items, defining prices to bins using the weights and
bin types, and finally using these prices to establish the lower
bound.

\paragraph{The assignment of weights to items.}
We assign weights to items as follows. For a $C_j$--item, where $2
\leq j \leq t-1$, we let its weight be $\frac{1}{7^{j-1}}$. The
weight of a $C_t$--item is $\frac{1}{6\cdot 7^{t-2}}$. The weight
of a large $A$--item is denoted by $w$ where $w \in [1,1.5]$. The
weight of any other item is $1$, those are $B_{11}$--items,
$B_{21}$--items, $B_{22}$--items, $B_{31}$--items,
$B_{32}$--items, and small $A$--items.

\paragraph{Definition of bin type.}
For a bin packed by the algorithm, we say that it has type $j$ if
it has a $C_j$--item for some $2 \leq j \leq t$ and no smaller
items (i.e., for any $k$ such that $j<k\leq t$, it has no
$C_k$--item). We say that it has type $1$ if it has an $A$--item
and no smaller items (i.e., it has no $C_k$--item  for all $2 \leq
k \leq t$). We say that it is a double bin if it has a
$B_{21}$--item or a $B_{31}$--item and no smaller items (i.e., no
$C_k$--item for all $2\leq k \leq t$ and no $A$--item), and we say
that it is a single bin if it has only items of sizes above $\frac
12$, i.e., a $B_{11}$--item or a $B_{22}$--item or a
$B_{32}$--item (where every such bin has exactly one item).

\paragraph{The price of a bin type.}
We define the price of a bin type as follows. A bin $D$ of a
certain type may receive additional items after its first batch of
items out of which its first item comes. Moreover, its contents
may differ in different continuations of the input (due to
branching). Consider the contents of $D$ for all continuations
simultaneously (taking into account the situation where these
items indeed arrive), and define a set of items $S(D)$ based on
this (one can think of $S(D)$ as a virtual bin, which is valid for
any possible input). For example, if the bin has one (large)
$A$--item, and in the first continuation it will receive a
$B_{11}$--item, in the second continuation it will receive one
$B_{21}$--item and one $B_{22}$--item, and in the third
continuation it would receive two $B_{31}$--items, then the set
$S(D)$ contains six items (one of size approximately $\frac 17$,
one of size approximately $\frac 13$, two of sizes approximately
$\frac 12$, and two of sizes approximately $\frac 5{14}$). The
price of $D$ is defined as the total weight of items of $S(D)$
(for the example, this price is $w+5$). The price of a bin type is
the supremum price of any bin of this type.

\paragraph{Calculating the prices of the bin types.}
Let $W_j$ denote the
price of bin type $j$, for $1 \leq j \leq t$, let $W_d$ denote the
price of a double bin, and let $W_s$ denote the price of a single
bin.

\begin{lemma}
For the weights defined above, we have $W_s=1$, $W_d=2$, and
$W_1=w+5$.
\end{lemma}
\begin{proof}
We have $W_s=1$ and $W_d=2$ as all items of sizes above $\frac 13$
have weights of $1$. This holds as a single bin has exactly one
item, while for a double bin $D$, $|S(D)|\leq 2$ and there is just
one continuation to be considered (so it either has two identical
items, or its second item has size above $\frac 12$, and in both
cases the price is $2$).

Consider now type $1$ bins. For such a bin $D_1$, we consider
$S(D_1)$. This is a bin whose first item is an $A$--item, it has
one large $A$--item and possibly also small $A$--items. The weight
of its large $A$--item is $w$, and we calculate the weight of
other items, and show that it is at most $5$. The number of small
$A$--items of $S(D_1)$ is between zero and five (as the sizes of
$A$--items are above $\frac 17$).

\begin{itemize}

\item  If $S(D_1)$ has at least four small $A$--items, it has no
space for further items (as the total size would be above $\frac
57+\frac 13>1$). In this case, $|S(D_1)| \leq 6$, and its price is
at most $w+5$.

\item If $S(D_1)$ has exactly three small $A$--items, the space
for other items is below $\frac 12$, so $S(D_1)$ can have one
$B_{21}$--item and one $B_{31}$--item, and its price is again
$w+5$.

\item If $S(D_1)$ has exactly two small $A$--items, $S(D_1)$ can
have one item for every continuation, and $|S(D_1)|\leq 6$, so its
price is again $w+5$.

\item If $S(D_1)$ has exactly one small $A$--item,  $S(D_1)$ can
have one $B_{31}$--item (or one $B_{32}$--item), but it cannot
have two $B_{31}$--items, as the size of one small $A$--item, one
large $A$--item, and two $B_{31}$--items is above
$(\frac{1+\eps}7)+(\frac{1+\eps+4\gam}7)+2(\frac{5-2\eps-3\gam}{14})=\frac{1+\eps+1+\eps+4\gam+5-2\eps-3\gam}7=\frac{7+\gam}7>1$.
It can contain a $B_{11}$--item, and it can contain two
$B_{21}$--items. Thus, $|S(D_1)| \leq 6$, so its price is again
$w+5$.

\item Finally, if it has no small $A$--items, it can contain two
$B_{31}$--items or one $B_{11}$--item or two $B_{21}$--items, and
in this case we also have $|S(D_1)| \leq 6$, and a price of $w+5$.
\end{itemize}
\end{proof}

\begin{lemma}
For the weights defined above, we have $W_j=7-\frac{1}{7^{j-1}}$
for $2 \leq j \leq t-1$, and $W_t=7$.
\end{lemma}
\begin{proof}
Consider a type $j$ bin $D_j$ (for $2 \leq j \leq t$). As an
$A$--item can be large only in the case that it is packed into an
empty bin, $S(D_j)$ may only have small $A$--items (and items
which are not $A$--items).

In the case $j=t$, we claim that every $C_k$--item for any $2 \leq
k \leq t-1$ can be replaced with $6\cdot 7^{t-k-1}$ $C_t$--items.
The total size of items is not increased (the size of a $C_k$ item
is above $\frac{1}{7^{k-1}}$ while the total size of $6\cdot
7^{t-k-1}$ $C_t$--items is below $\frac{6\cdot 7^{t-k-1}}{6\cdot
7^{t-1}}=\frac{1}{7^k}$, so for any possible continuation after
the $A$--items, the remaining members of $S(D_t)$ can still be
packed. Similarly, every (small) $A$--item is replaced with
$6\cdot 7^t$ $C_t$--items (whose total size is smaller). The
weight is unchanged since the weight of every $C_t$--item is
$\frac{1}{6\cdot 7^{t-2}}$, so the total weight of $6\cdot
7^{t-k-1}$ $C_t$--items is $\frac{1}{7^{k-1}}$, which is the
weight of a $C_k$--item, and this is valid for the case $k=1$ (of
$A$--items) as well. Thus, we assume that $S(D_t)$ has some number
of $C_t$--items, and calculate the number of $B_{11}$--items,
$B_{21}$--items, and $B_{31}$--items that can be packed given the
number of $C_t$--items (as $B_{22}>B_{21}$ and $B_{32}>B_{31}$,
and they all have weights of $1$, we consider only single items
and pairs of $B_{21}$--items and pairs of $B_{31}$--items).

We show what cannot be included in $S(D_t)$. The maximum number of
$C_t$--items is $6\cdot 7^{t-1}=42\cdot 7^{t-2}$. We will use the
following property: $(28 \cdot 7^{t-2}+1)\cdot 294\eps+\eps <
\frac{1}{6 \cdot 7^{t-1}}$. This property holds as $\eps <
\frac{1}{6^t\cdot 7^{3t}}$ and therefore $$(28\cdot
7^{t-2}+1)\cdot 294\eps+\eps \leq (24\cdot 7^{t+1}+295)\eps <
25\cdot 7^{t+1}\eps < \frac{25\cdot 7^{t+1}}{6^t\cdot
7^{3t}}<\frac{1}{6^{t-1}\cdot 7^{2t-2}} \ , $$ as $t \geq 3$.

\begin{itemize}

\item If there are at least $12\cdot 7^{t-2}+1$ $C_t$--items in
$S(D_t)$, there cannot be two $B_{31}$--items as $(12\cdot
7^{t-2}+1)(\frac{1}{6\cdot
7^{t-1}}-294\eps)+2(\frac{5-2\eps-3\gam}{14}) \geq
1+\frac{1}{6\cdot 7^{t-1}} - (12\cdot 7^{t-2}+1)294\eps-\eps >1$.

\item If there are at least $14\cdot 7^{t-2}+1$ $C_t$--items in
$S(D_t)$, there cannot be two $B_{21}$--items as $(14\cdot
7^{t-2}+1)(\frac{1}{6\cdot 7^{t-1}}-294\eps)+2(\frac{1+\eps}{3}) >
1+\frac{1}{6\cdot 7^{t-1}} -(14\cdot 7^{t-2}+1)294\eps >1$.

\item If there are at least $21\cdot 7^{t-2}+1$ $C_t$--items in
$S(D_t)$, there cannot be a $B_{11}$--item as $(21\cdot
7^{t-2}+1)(\frac{1}{6\cdot 7^{t-1}}-294\eps)+(\frac{1+2\eps}{2})
> 1+\frac{1}{6\cdot 7^{t-1}} -(21\cdot 7^{t-2}+1)294\eps >1$.

\item If there are at least $28\cdot 7^{t-2}+1$ $C_t$--items in
$S(D_t)$, there cannot be a $B_{21}$--item and there cannot be a
$B_{31}$--item as $(28 \cdot 7^{t-2}+1)(\frac{1}{6\cdot
7^{t-1}}-294\eps)+(\frac{1+\eps}{3})
> 1+\frac{1}{6\cdot 7^{t-1}} -(28\cdot 7^{t-2}+1)294\eps >1$ and the fact that a $B_{31}$ item is larger than a $B_{21}$--item.
\end{itemize}

Now, we can find upper bounds on the prices in all cases.

\begin{itemize}
\item  If the number of $C_t$--items is at most $12\cdot 7^{t-2}$,
the price is at most $5+\frac{12\cdot 7^{t-2}}{6\cdot 7^{t-2}}=7$.

\item If the number of $C_t$--items is at least $12\cdot
7^{t-2}+1$ and at most $14\cdot 7^{t-2}$,  the price is at most
$4+\frac{14\cdot 7^{t-2}}{6\cdot 7^{t-2}} < 7$.

\item If the number of $C_t$--items is at least $14\cdot
7^{t-2}+1$ and at most $21\cdot 7^{t-2}$,  the price is at most
$3+\frac{21\cdot 7^{t-2}}{6\cdot 7^{t-2}} < 7$.

\item If the number of $C_t$--items is at least $21\cdot
7^{t-2}+1$ and at most $28\cdot 7^{t-2}$,  the price is at most
$2+\frac{28\cdot 7^{t-2}}{6\cdot 7^{t-2}} < 7$.

\item If the number of $C_t$--items is at least $28\cdot
7^{t-2}+1$, the price is at most $\frac{42 \cdot 7^{t-2}}{6\cdot
7^{t-2}} = 7$.
\end{itemize}

Next, we consider $D_k$ for $k<t$, and show that the price is
slightly smaller. No bin can contain more than ${7^k-1}$
$C_k$--items (as their sizes are above $\frac 1{7^k}$).  Here, we
can replace every item of size $C_j$ (for $2\leq j<k$) by exactly
$7^{k-j}$ $C_k$--items without modifying the total weight and
similarly we can replace every small $A$ item by $7^{k-1}$ $C_k$
items without changing the total weight.  Thus, we will assume
that $D_k$ does not contain such items. We will use the properties
that the numbers $7^k-1$ and $7^k+2$ are divisible by $3$, and the
numbers $7^k-1$ and $7^k+1$ are divisible by $2$.

\begin{itemize}

\item If there are at least $2\cdot 7^{k-1}$ $C_k$--items in
$S(D_k)$, there cannot be two $B_{31}$--items as $(2\cdot
7^{k-1})(\frac{1+28\eps}{ 7^{k}})+2(\frac{5-2\eps-3\gam}{14}) >
1+7\eps$.

\item  If there are at least $(7^{k}+2)/3$ $C_k$--items in
$S(D_k)$, there cannot be two $B_{21}$--items as
$((7^{k}+2)/3)(\frac{1+28\eps}{7^{k}})+2(\frac{1+\eps}{3}) > 1$.

\item If there are at least $(7^{k}+1)/2$ $C_k$--items in
$S(D_k)$, there cannot be a $B_{11}$--item as
$((7^{k}+1)/2)(\frac{1+28\eps}{7^{k}})+(\frac{1+2\eps}{2}) > 1$.

\item If there are at least $(9\cdot 7^{k-1}+1)/2$ $C_k$--items in
$S(D_k)$, there cannot be a $B_{31}$--item as $((9\cdot
7^{k-1}+1)/2)(\frac{1+28\eps}{7^{k}})+(\frac{5-2\eps-3\gam}{14})
\geq 1+\frac 1{2\cdot 7^k}-\eps>1$.

\item If there are at least $(2\cdot 7^{k}+1)/3$ $C_k$--items in
$S(D_k)$, there cannot be a $B_{21}$--item as $((2\cdot
7^{k}+1)/3)(\frac{1+28\eps}{ 7^{k}})+(\frac{1+\eps}{3})
> 1$.
\end{itemize}

Now, we can find upper bounds on the price in all cases.

\begin{itemize}
\item If the number of $C_k$--items is at most $2\cdot 7^{k-1}-1$,
the price is at most $5+\frac{2\cdot
7^{k-1}-1}{7^{k-1}}=7-\frac{1}{7^{k-1}}=\frac{7^k-1}{7^{k-1}}$.

\item If the number of $C_k$--items is at least $2\cdot 7^{k-1}$
and at most $(7^{k}-1)/3$,  the price is at most
$4+\frac{(7^{k}-1)/3}{7^{k-1}} =\frac{19\cdot 7^{k-1}-1}{3\cdot
7^{k-1}} < \frac{21\cdot 7^{k-1}-3}{3\cdot
7^{k-1}}=\frac{7^k-1}{7^{k-1}}$, as $k \geq 2$.

\item If the number of $C_k$--items is at least $(7^{k}+2)/3$ and
at most $(7^{k}-1)/2$,  the price is at most
$3+\frac{(7^{k}-1)/2}{ 7^{k-1}} =\frac{13\cdot 7^{k-1}-1}{2\cdot
7^{k-1}} < \frac{14\cdot 7^{k-1}-2}{2 \cdot
7^{k-1}}=\frac{7^k-1}{7^{k-1}}$, as $k \geq 2$.

\item  If the number of $C_k$--items is at least $(7^{k}+1)/2$ and
at most $(9\cdot 7^{k-1}-1)/2$,  the price is at most
$2+\frac{(9\cdot 7^{k-1}-1)/2}{7^{k-1}} =\frac{13\cdot
7^{k-1}-1}{2\cdot 7^{k-1}} < \frac{7^k-1}{7^{k-1}}$.

\item If the number of $C_k$--items is at least $(9\cdot
7^{k-1}+1)/2$ and at most  $(2\cdot 7^{k}-2)/3$ ,  the price is at
most $1+\frac{(2\cdot 7^{k}-2)/3}{7^{k-1}}  =\frac{17\cdot
7^{k-1}-2}{3\cdot 7^{k-1}} < \frac{21\cdot 7^{k-1}-3}{3\cdot
7^{k-1}}=\frac{7^k-1}{7^{k-1}}$, as $k \geq 2$.

\item  If the number of $C_k$--items is at least $(2\cdot
7^{k}+1)/3$, the price is at most $\frac{7^{k}-1}{7^{k-1}}$.
\end{itemize}

This concludes the proof.
\end{proof}

\paragraph{Using the prices of bin types to establish the lower bound on the asymptotic competitive ratio.}

Let $\nu_j$ denote the number of bins opened for $C_j$--items
(bins used for the first time when the batch of $C_j$--items is
presented). Let $\nu_1$ denote the number of bins opened for
$A$--items. Let $\nu_{k\ell}$ denote the number of bins opened for
$B_{k\ell}$--items, for $(k,\ell)\in ISB$, where $ISB
=\{(1,1),(2,1),(2,2),(3,1),(3,2)\}$.  Moreover, as large
$A$--items are exactly those $A$--items that are packed as first
items of their bins, we have $\nu_1=n_L$.

Let $ALG_j$ denote the cost of the algorithm for the input up to
the batch of $C_j$--items, and let $ALG_1$ denote the cost of the
algorithm up to the batch of $A$--items. Let $ALG_{k\ell}$ denote
the cost of the algorithm up to the batch of $B_{k\ell}$--items
for $(k,\ell)\in ISB$.

Let $R$ be the asymptotic competitive ratio of ALG, and let $f$ be
a function such that $f(n)=o(n)$ and for any input $I$ it holds
that $ALG(I) \leq R \cdot OPT(I)+f(OPT(I))$.

We have $ALG_j \leq R \cdot OPT_j+f(OPT_j)$ for $1 \leq j \leq t$.
We also have $ALG_{k\ell} \leq R \cdot OPT_{k\ell}+f(OPT_{k\ell})$
for $(k,\ell)\in ISB$.

Let $W$ denote the total weight of all items (for all branches,
such that every possible item is counted exactly once). Since
$\frac{1}{6\cdot 7^{t-2}}+\sum_{j=2}^{t-1} \frac 1{7^{j-1}} =
\frac 16$, we have
\begin{eqnarray*}W&=&N\cdot
(\frac{1}{6\cdot 7^{t-2}}+\sum_{j=2}^{t-1} \frac 1{7^{j-1}})+
w\cdot n_L+ (N-n_L)+\frac N3  +2N +n_{31}+n_{32}\\ &=&\frac N6 +
(w-1)n_L+\frac{10N}3 +\frac{7N - 7n_L}6 + \frac{7N -5n_L}6 \\ &=& w \cdot
n_L-3 \cdot n_L + \frac{35N}6 \ .
\end{eqnarray*}

\begin{lemma}\label{fiv}
We have $W \leq \sum_{j=1}^{t}  W_j\nu_j +
W_d(\nu_{21}+\nu_{31})+W_s(\nu_{11}+\nu_{22}+\nu_{32})=\sum_{j=1}^{t}
W_j\nu_j+\nu_{11}+2\nu_{21}+\nu_{22}+2\nu_{31}+\nu_{32}$.
\end{lemma}
\begin{proof}
The weight of every item is included in the price of exactly one bin used by the
algorithm. Thus, the total weight is equal to the total price of
bins. Given the supremum prices, we get an upper bound on the
total price.  This proves the inequality, the equality holds by substituting the values of $W_d$ and $W_s$.
\end{proof}

Let $n'_L=\frac{n_L}{N}$, and $W'=\frac{W}N=w \cdot n'_L-3 \cdot
n'_L + \frac{35}6$.

\begin{lemma}\label{finallem}
For any value of $n_L$ ($0 \leq n_L \leq N$) and for any value of
$w$ ($1 \leq w \leq 1.5$), we have $R \geq
\frac{W'}{2133/588-1.25n'_L+\frac{1}{7\cdot 48\cdot49^{t-2}}+\frac{1}{48\cdot
49}+w/7}$, and therefore $R \geq \frac{w\cdot n'_L-3 \cdot n'_L +
\frac{35}6}{8533/2352-1.25n'_L+w/7}$.
\end{lemma}
\begin{proof}
As the optimal costs are not smaller than $M$, and $M$ can be
chosen to be sufficiently large, we will neglect the additive term
of $f(OPT(I))$, and assume that for every input $I$ for which $OPT(I) \geq M$, we have
$ALG(I) \leq R \cdot OPT(I)$.

We will write the constraints for all possible inputs (with all
stopping points and continuations), and we will take a linear
combination of them using positive multipliers. For an input $I$,
we will exhibit a formula for $ALG(I)$ and an upper bound for
$OPT(I)$.  The inequality for this input is that the formula for
$ALG(I)$ is at most $R$ times the upper bound for $OPT(I)$.  This
inequality is the one we multiply by the corresponding multiplier.

For $1 \leq j \leq t$ we have $ALG_j=\sum_{i=j}^t \nu_i$ and
$OPT_j \leq \frac{N}{6 \cdot 7^{j-1}}$. The multiplier for $j\geq
2$ is $W_j-W_{j-1}$, and the multiplier for $j=1$ is
$W_1-W_{s}-2\cdot W_{d}=5+w-5=w$. For $j=t$, we have
$W_t-W_{t-1}=7-(7-\frac{1}{7^{t-2}})=\frac{1}{7^{t-2}}$. For $3
\leq j \leq t-1$, we have
$W_j-W_{j-1}=(7-\frac{1}{7^{j-1}})-(7-\frac{1}{7^{j-2}})=\frac{1}{7^{j-2}}-\frac{1}{7^{j-1}}=\frac{6}{7^{j-1}}$.
For $j=2$, we have
$W_2-W_{1}=(7-\frac{1}{7})-(5+w)=\frac{13}7-w>0$, as $w \leq 1.5$.
Let $\Delta=\sum_{j=1}^t \nu_j$.

For the input that ends with $B_{11}$--items,
$ALG_{11}=\Delta+\nu_{11}$ while $OPT_{11} \leq \frac{N}3$. The
multiplier is $W_s=1$.

For the input that ends with $B_{21}$--items,
$ALG_{21}=\Delta+\nu_{21}$ while $OPT_{21} \leq \frac{N}2$. The
multiplier is $W_d-W_s=1$.

For the input that ends with $B_{22}$--items,
$ALG_{22}=\Delta+\nu_{21}+\nu_{22}$ while $OPT_{22} \leq N$. The
multiplier is $W_s=1$.

For the input that ends with $B_{31}$--items,
$ALG_{31}=\Delta+\nu_{31}$.  In this case $OPT_{31} \leq
\frac{7N-5n_L}{12}$. The multiplier is $W_d-W_s=1$.

For the input that ends with $B_{32}$--items,
$ALG_{32}=\Delta+\nu_{31}+\nu_{32}$.  In this case $OPT_{32} \leq
\frac{7N  -5n_L}6$. The multiplier is $W_s=1$.

Taking the sum of these inequalities (multiplied by the chosen
multipliers) we have a left hand side of
\begin{eqnarray*}
\displaystyle
&&w\Delta+5\Delta+\nu_{11}+2\nu_{21}+\nu_{22}+2\nu_{31}+\nu_{32}+\sum_{j=2}^t
(W_j-W_{j-1})(\sum_{i=j}^t \nu_i) \\&=&
\nu_{11}+2\nu_{21}+\nu_{22}+2\nu_{31}+\nu_{32}+(5+w)\Delta+\sum_{j=2}^t
(W_j-W_1)\nu_j  \\ & =&
\nu_{11}+2\nu_{21}+\nu_{22}+2\nu_{31}+\nu_{32}+(5+w)\Delta+\sum_{j=2}^t
W_j\nu_j-W_1\sum_{j=1}^t \nu_j+W_1\nu_1
\\&=& \nu_{11}+2\nu_{21}+\nu_{22}+2\nu_{31}+\nu_{32}+(5+w)\Delta+\sum_{j=1}^t
W_j\nu_j-W_1\Delta\\&=& \sum_{j=1}^{t}
W_j\nu_j+\nu_{11}+2\nu_{21}+\nu_{22}+2\nu_{31}+\nu_{32} \geq W \ ,
\end{eqnarray*} where the last inequality holds
by Lemma \ref{fiv}.

The right hand side is $R$ multiplied by
$$\frac{11N}6+\frac{7N-5n_L}{12}+\frac{7N
-5n_L}6+\frac{1}{7^{t-2}}\cdot \frac{N}{6\cdot
7^{t-1}}+\sum_{j=3}^{t-1} \frac{6}{7^{j-1}}\cdot\frac{N}{6\cdot
7^{j-1}}+(\frac{13}7-w)\frac{N}{42}+w\frac
N6$$$$=\frac{2133N}{588}-\frac{5n_L}4+\frac{N}{7\cdot 48\cdot
49^{t-2}}+\frac{N}{48\cdot 49}+\frac{w\cdot N}{7} \ . $$ Thus, by
the resulting inequality we deduce the first lower bound on $R$.
The second inequality (in the statement of the lemma) holds as the
first one holds for all integers $t\geq 3$ and by letting $t$ to
grow unbounded, we establish the second lower bound on $R$ from
the first one.
\end{proof}

\begin{theorem}
We have $R \geq \frac{1363-\sqrt{1387369}}{120} \approx
1.5427809064729$.  That is, there is no online algorithm for bin packing with asymptotic competitive ratio strictly smaller than $ \frac{1363-\sqrt{1387369}}{120} \approx
1.5427809064729$.
\end{theorem}
\begin{proof}
Let $r=\frac{1363-\sqrt{1387369}}{120}\approx 1.5427809064729$ and
let $w=\frac{\sqrt{1387369}-1075}{96}\approx  1.07152386690879$,
where $w=3-1.25\cdot r$.

We have $R \geq \frac{w\cdot n'_L-3 \cdot n'_L +
\frac{35}6}{8533/2352-1.25n'_L+w/7}$, and we show that this
expression is equal to $r$ (for any $n'_L$, where $0 \leq n'_L
\leq 1$). The denominator is positive as $8533/2352-1.25n'_L+w/7 >
8533/2352-1.25+1/7 >2$, by $n'_L\leq 1$ and $w \geq  1$. Thus, it
is equivalent to showing $w\cdot n'_L-3 \cdot n'_L + \frac{35}6 =
r(8533/2352-1.25n'_L+w/7)$, which is equivalent to
$n'_L(w-3+1.25\cdot r)+ \frac{35}6 - r(8533/2352+w/7)=0$.

Indeed $w-3+1.25\cdot r =0$, by the choice of $w$ and $r$.
Additionally, $\frac{35}6 - r(8533/2352+w/7)=\frac{35}6 -
r(8533/2352+(3-1.25\cdot r)/7)=0$, by the choice of $r$.
\end{proof}

\begin{remark}
We note that our choice of $w$ and $r$ are optimal in the sense
that the lower bound of Lemma \ref{finallem} cannot be used to
prove a higher lower bound on $R$ using other values of $w$ for
the formula which we obtained. This can be observed by solving the
corresponding mathematical program of maximizing (over the
possible values of $w$) of minimizing (over the possible values of
$n'_L$) of the ratio function defined using these two parameters
that we establish in Lemma \ref{finallem}.
\end{remark}

\bibliographystyle{abbrv}
\bibliography{bplb}

\begin{thebibliography}{10}

\bibitem{BCKK04}
L.~Babel, B.~Chen, H.~Kellerer, and V.~Kotov.
\newblock Algorithms for on-line bin-packing problems with cardinality
  constraints.
\newblock {\em Discrete Applied Mathematics}, 143(1-3):238--251, 2004.

\bibitem{DBLP:journals/corr/BaloghBDEL17}
J.~Balogh, J.~B{\'{e}}k{\'{e}}si, Gy.~D{\'{o}}sa, L.~Epstein, and A.~Levin.
\newblock A new and improved algorithm for online bin packing.
\newblock {\em CoRR}, abs/1707.01728, 2017.
\newblock Also in ESA 2018, to appear.

\bibitem{BBDEL_ESA}
J.~Balogh, J.~B{\'{e}}k{\'{e}}si, Gy.~D{\'{o}}sa, L.~Epstein, and A.~Levin.
\newblock Online bin packing with cardinality constraints resolved.
\newblock In {\em Proc. of the 25th European Symposium on Algorithms
  (ESA2017)}, pages 10:1--10:14, 2017.

\bibitem{BBDGT15}
J.~Balogh, J.~B\'{e}k\'{e}si, Gy.~D\'{o}sa, G.~Galambos, and Z.~Tan.
\newblock Lower bound for 3-batched bin packing.
\newblock {\em Discrete Optimization}, 21:14--24, 2016.

\bibitem{balogh2012new}
J.~Balogh, J.~B{\'e}k{\'e}si, and G.~Galambos.
\newblock New lower bounds for certain classes of bin packing algorithms.
\newblock {\em Theoretical Computer Science}, 440:1--13, 2012.

\bibitem{BDE}
J.~B{\'e}k{\'e}si, Gy.~D\'osa, and L.~Epstein.
\newblock Bounds for online bin packing with cardinality constraints.
\newblock {\em Information and Computation}, 249:190--204, 2016.

\bibitem{Blitz}
D.~Blitz.
\newblock Lower bounds on the asymptotic worst-case ratios of on-line bin
  packing algorithms.
\newblock Technical Report 114682, University of Rotterdam, 1996.
\newblock M.Sc. thesis.

\bibitem{brown1979lower}
D.~J. Brown.
\newblock A lower bound for on-line one-dimensional bin packing algorithms.
\newblock {\em Coordinated Science Laboratory Report no. R-864 (UILU-ENG
  78-2257)}, 1979.

\bibitem{Gon32}
E.~G. {Coffman {Jr.}} and J.~Csirik.
\newblock Performance guarantees for one-dimensional bin packing.
\newblock In T.~F. Gonzalez, editor, {\em Handbook of Approximation Algorithms
  and Metaheuristics}, chapter~32, pages (32--1)--(32--18). Chapman \&
  Hall/Crc, 2007.

\bibitem{CsiWoe98}
J.~Csirik and G.~J. Woeginger.
\newblock On-line packing and covering problems.
\newblock In {\em A.~Fiat and G.~J. Woeginger, editors, {\em Online Algorithms:
  {The} State of the Art}}, pages 147--177, 1998.

\bibitem{FK13}
H.~Fujiwara and K.~M. Kobayashi.
\newblock Improved lower bounds for the online bin packing problem with
  cardinality constraints.
\newblock {\em Journal of Combinatorial Optimization}, 29(1):67--87, 2015.

\bibitem{HvS16}
S.~Heydrich and R.~van Stee.
\newblock Beating the harmonic lower bound for online bin packing.
\newblock In {\em Proc. of 43rd International Colloquium on Automata,
  Languages, and Programming (ICALP2016)}, pages 41:1--41:14, 2016.

\bibitem{J73}
D.~S. Johnson.
\newblock {\em Near-optimal bin packing algorithms}.
\newblock PhD thesis, MIT, Cambridge, MA, 1973.

\bibitem{J74}
D.~S. Johnson.
\newblock Fast algorithms for bin packing.
\newblock {\em Journal of Computer and System Sciences}, 8:272--314, 1974.

\bibitem{JDUGG74}
D.~S. Johnson, A.~Demers, J.~D. Ullman, M.~R. Garey, and R.~L. Graham.
\newblock Worst-case performance bounds for simple one-dimensional packing
  algorithms.
\newblock {\em SIAM Journal on Computing}, 3:256--278, 1974.

\bibitem{LeeLee85}
C.~C. Lee and D.~T. Lee.
\newblock A simple online bin packing algorithm.
\newblock {\em Journal of the ACM}, 32(3):562--572, 1985.

\bibitem{liang1980lower}
F.~M. Liang.
\newblock A lower bound for on-line bin packing.
\newblock {\em Information Processing Letters}, 10(2):76--79, 1980.

\bibitem{RaBrLL89}
P.~Ramanan, D.~J. Brown, C.~C. Lee, and D.~T. Lee.
\newblock Online bin packing in linear time.
\newblock {\em Journal of Algorithms}, 10:305--326, 1989.

\bibitem{Seiden02J}
S.~S. Seiden.
\newblock {On the online bin packing problem}.
\newblock {\em Journal of the ACM}, 49(5):640--671, 2002.

\bibitem{U71}
J.~D. Ullman.
\newblock The performance of a memory allocation algorithm.
\newblock Technical Report 100, Princeton University, Princeton, NJ, 1971.

\bibitem{Vliet92}
A.~van Vliet.
\newblock An improved lower bound for online bin packing algorithms.
\newblock {\em Information Processing Letters}, 43(5):277--284, 1992.

\bibitem{Yao80A}
A.~C.~C. Yao.
\newblock New algorithms for bin packing.
\newblock {\em Journal of the {ACM}}, 27:207--227, 1980.

\end{thebibliography}

\end{document}